%% file: paper.tex
\documentclass{llncs}
\pdfoutput=1

\usepackage{times}

\usepackage{amssymb}
\usepackage{amsfonts}
\usepackage{amsmath}
\usepackage{graphicx}
\usepackage{tabularx}
\usepackage[noend]{algorithmic}
\usepackage{subfigure}
\usepackage{rotating}
\usepackage{multirow}
\usepackage{picinpar}

\input{macros}

\title{SMT-based Induction Methods for Timed Systems}
\author{Roland Kindermann \and Tommi Junttila \and Ilkka Niemel{\"a}}
\institute{Aalto University\\
Department of Information and Computer Science\\
P.O.Box 15400, FI-00076 Aalto, Finland\\
\email{Roland.Kindermann@aalto.fi, Tommi.Junttila@aalto.fi, Ilkka.Niemela@aalto.fi}%
}

\begin{document}

\maketitle
\begin{abstract} 
Modeling time related aspects is important in many applications of verification methods.
For precise results, it is necessary to interpret time as a dense domain, e.g. using timed automata as a formalism, even though the system's resulting infinite state space is challenging for verification methods.
Furthermore, fully symbolic treatment of both timing related and non-timing related elements of the state space seems to offer an attractive approach to model checking timed systems with a large amount of non-determinism.
This paper presents an SMT-based timed system extension to the IC3 algorithm, a SAT-based novel, highly efficient, complete verification method for untimed systems.
Handling of the infinite state spaces of timed system in the extended IC3 algorithm is based on suitably adapting the well-known region abstraction for timed systems.
Additionally, $k$-induction,
another symbolic verification method for discrete time systems,
is extended in a similar fashion to support timed systems.
Both new methods are evaluated and experimentally compared to
a booleanization-based verification approach that
uses the original discrete time IC3 algorithm.
\end{abstract}

\section{Introduction} 

\input{introduction}



\section{Symbolic Timed Transition Systems and Regions} 

\label{s:stts}
\label{s:regions}

\input{stts}


\section{$k$-Induction for Timed Systems} 

\label{s:k-ind}

\input{k-ind}


\section{IC3 for Timed Systems} 

\label{s:ic3}

\input{tic3}


\section{Optimizations by Excluding Multiple Regions} 

\label{s:opt}

\input{optimizations}


\section{Experiments} 

\input{experiments}


\section{Conclusion} 


\input{conclusion}

\bibliographystyle{splncs}
\bibliography{paper}

\end{document}

%% file: macros.tex
\newcommand{\DelayTwo}{\eta}

\newcommand{\Def}{\mathrel{:=}}
\newcommand{\Assign}{\mathrel{:=}}

\newcommand{\Implies}{\Rightarrow}
\newcommand{\Iff}{\Leftrightarrow}

\newcommand{\True}{\mathbf{true}}
\newcommand{\False}{\mathbf{false}}
\newcommand{\Tuple}[1]{{\langle{#1}\rangle}}

\newcommand{\RealsPos}{\mathbb{R}_{+}}
\newcommand{\RealsNonNeg}{\mathbb{R}_{\geq 0}}
\newcommand{\Integers}{\mathbb{Z}}
\newcommand{\IntegersPos}{\mathbb{N}}
\newcommand{\Naturals}{\mathbb{N}}
\newcommand{\Set}[1]{\{#1\}}
\newcommand{\Models}[2]{{{#2} \models {#1}}}

\newcommand{\Next}[1]{#1'}
\newcommand{\System}{S}
\newcommand{\Init}{\mathit{Init}}
\newcommand{\Einit}{\widehat{\mathit{Init}}}
\newcommand{\Inv}{\mathit{Invar}}
\newcommand{\Invar}{\Inv}
\newcommand{\Einvar}{\widehat{\mathit{Invar}}}
\newcommand{\ESystem}{\widehat{\System}}
\newcommand{\Trans}{\ensuremath{T}}
\newcommand{\Etrans}{\widehat{T}}

\newcommand{\SVar}{x}
\newcommand{\SVarNext}{\Next{x}}
\newcommand{\SVars}{X}
\newcommand{\SVarsNext}{\Next{X}}
\newcommand{\Clock}{c}
\newcommand{\AnotherClock}{d}
\newcommand{\FixedClock}{\hat{c}}
\newcommand{\Clocks}{C}
\newcommand{\ClocksNext}{\Next{C}}
\newcommand{\Reset}[1]{r_{#1}}
\newcommand{\Resets}{R}
\newcommand{\Urgent}{U}
\newcommand{\STTS}{\Tuple{\SVars,\Clocks,\Init,\Inv,\Trans,\Resets}}
\newcommand{\Utinit}{\Init}
\newcommand{\Uttrans}{\Trans}
\newcommand{\Utinvar}{\Inv}
\newcommand{\UntimedSys}{\Tuple{\SVars,\emptyset,\Utinit,\Utinvar,\Uttrans,\emptyset}}
\newcommand{\ExpsttsTerm}{explicit STTS}
\newcommand{\aExpsttsTerm}{an explicit STTS}

\newcommand{\Expstts}{\Tuple{\SVars,\Clocks,\Einit,\Einvar,\Etrans}}
\newcommand{\Property}{P}
\newcommand{\Deltavar}{\delta}
\newcommand{\ForAllClocks}{\bigwedge_{\Clock\in\Clocks}}

\newcommand{\ForAnyClock}{\bigvee_{\Clock\in\Clocks}}
\newcommand{\ForAnyClockPair}{\bigvee_{\Clock\in\Clocks}\bigvee_{\AnotherClock\in\Clocks\setminus\Set{\Clock}}}
\newcommand{\ClockMax}[1]{m_{#1}}

\newcommand{\SameRegion}{\sim}
\newcommand{\NotSameRegion}{\not\sim}
\newcommand{\ClockValuation}{v}
\newcommand{\AnotherClockValuation}{w}
\newcommand{\EClass}[1]{{[#1]}}

\newcommand{\Index}{i}
\newcommand{\AnotherIndex}{j}
\newcommand{\Constant}{n}

\newcommand{\K}{k}
\newcommand{\Int}[1]{{#1}_\mathrm{int}}
\newcommand{\Frac}[1]{{#1}_\mathrm{fract}}
\newcommand{\AtMax}[1]{{max}_{#1}}
\newcommand{\AtStep}[2]{{#2^{[#1]}}}
\newcommand{\SameRegionState}{t}
\newcommand{\SeqState}[1]{t_{#1}}

\newcommand{\SeqFirst}{\SeqState{0}}
\newcommand{\SeqLast}{\SeqState{n}}
\newcommand{\SameRegionSeq}[1]{#1'}
\newcommand{\ThirdIndexKIndOptProof}{k}

\newcommand{\FSetNoIndex}{F}
\newcommand{\FSet}[1]{{\FSetNoIndex}_{#1}}
\newcommand{\ICTProof}{{Proof}}
\newcommand{\FDepth}{k}
\newcommand{\Asgn}{\leftarrow}
\newcommand{\Cube}{s}
\newcommand{\AnotherCube}{z}
\newcommand{\BlockFunc}{{blockState}}
\newcommand{\Queue}{Q}
\newcommand{\ClockValue}[1]{s(#1)}
\newcommand{\CFrac}[1]{\operatorname{fract(#1)}}
\newcommand{\CInt}[1]{\lfloor{#1}\rfloor}
\newcommand{\FracValue}[1]{\CFrac{\ClockValue{#1}}}
\newcommand{\IntValue}[1]{\CInt{\ClockValue{#1}}}

\newcommand{\ClockOrSVal}{y}

\newcommand{\Pred}{p}
\newcommand{\Delay}{\delta}
\newcommand{\NumSteps}{n}

\newcommand{\TPrec}{\precsim}

\newcommand{\Ceil}[1]{{\lceil{#1}\rceil}}
\newcommand{\Floor}[1]{{\lfloor{#1}\rfloor}}

\newcommand{\XRTimerb}{T_d}
\newcommand{\XTimera}{c}
\newcommand{\XTimerb}{d}
\newcommand{\XInput}{x_1}
\newcommand{\XOutput}{x_2}
\newcommand{\Th}{\mathcal{T}}
\newcommand{\Interp}{\mathcal{I}}
\newcommand{\Form}{\phi}
\newcommand{\Setdef}[2]{{\left\{{#1}\mid{#2}\right\}}}
\newcommand{\Eval}[2]{{#2}(#1)}
\newcommand{\Cubep}{u}
\newcommand{\Succ}{\longrightarrow}

%% file: introduction.tex
In many application areas of model checking,
such as analysis of safety instrumented systems,
modeling and analyzing in the presence of
dense time constructions such as timers and delays is essential.
Compared to finite state systems,
such timed systems add an extra layer of challenge for model checking tools.
In many cases,
timed automata~\cite{DBLP:journals/tcs/AlurD94,DBLP:conf/cav/Alur99,DBLP:conf/ac/BengtssonY03}
are a convenient formalism for describing and model checking timed systems.
There are many tools,
Uppaal \cite{Uppaaltutorial:2004} to name just one, for timed automata
and
model checking algorithms for timed automata have been studied extensively
during the last two decades,
see e.g.~\cite{DBLP:conf/ac/BengtssonY03} for an overview.
Most state-of-the-art model checking systems for timed automata
use the so-called region abstraction to make a finite state
abstraction of the dense time clocks in the automata.
These regions are then manipulated symbolically with
difference bounded matrices or
decision diagram structures (see e.g.~\cite{Wang:STTT2004}).

In this paper our focus is on model checking of safety
instrumented systems (see e.g.~\cite{GruhnCheddie2006}).
Such systems  have features that are challenging for the classic
timed automata based approach described above.
First, safety instrumented systems do typically involve a substantial number of
timing related issues.
However, such systems are often not best described using automata-like control
structures but with a sequential circuit-like control logic.
This makes the use of timed automata rather inconvenient in modeling.
Second,
such systems tend to have a relatively large amount of non-deterministic input
signals which are computationally challenging for model checking tools
based on explicit state representation of
discrete components (i.e.~control location and data).

Hence,  we are interested in developing model checking techniques that
complement the automata based methods to address these issues. 
Instead of timed automata,
we use a more generic symbolic system description formalism~\cite{KindermannJunttilaNiemela:ACSD2011}
which can be seen as an extension of
the classic symbolic transition systems~\cite{MannaPnueli:1992}
with dense time clock variables and constraints.
In our previous work~\cite{KindermannJunttilaNiemela:ACSD2011},
we have experimented with
(i) SMT-based bounded model checking (BMC)~\cite{DBLP:conf/forte/AudemardCKS02,DBLP:journals/entcs/Sorea02},
and
(ii) BDD-based model checking based on booleanization of the region abstracted
model.
These methods were not totally satisfactory as
(i) BMC can, in practice, only find bugs, not prove correctness of the system,
and
(ii) the BDD-based method does not seem to scale well to realistically sized
models.

In order to address the computational challenge to develop 
model checking techniques that can handle timing as well as
a substantial amount of non-deterministic input signals and
prove correctness,
we turn to inductive techniques.
The motivation here is the success of temporal induction~\cite{SheeranEtAl:FMCAD2000,EenSorensson:2003}
and, especially,
of the IC3 algorithm~\cite{DBLP:conf/vmcai/Bradley11}
in the verification of finite state hardware systems.
%
Our approach is to employ SMT solvers instead of SAT solvers as the
basic constraint solver technology and
apply symbolic region abstraction to handle the dense time clocks in the models.
We extend IC3 to timed systems
by using linear arithmetics instead of propositional logic
and
by lifting the concrete states found by the SMT solver to
symbolic region level constraints
that are further used in the subsequent steps to constrain the search.
As a result we obtain a version of IC3 that does not explicitly
construct the symbolic region abstracted system
but still can exclude whole regions of states at once.
We also describe an SMT-based extension of the $\K$-induction algorithm
to these kinds of timed systems.
In addition,
we develop optimizations that allow us
to exclude more regions at a time in the SMT-based IC3 algorithm,
and
to use stronger ``simple path'' constraints in $\K$-induction.


Our experimental results indicate that SMT-based IC3 can indeed prove
much more properties and on much larger models than
were possible with our earlier approaches or with SMT-based timed $\K$-induction.
Furthermore,
when comparing to the approach of using the original propositional IC3
on booleanized region abstracted model,
we observe that using richer logics in the SMT framework
makes the IC3 algorithm scale much better for timed systems.
However, IC3 seems to perform worse than $\K$-induction (and thus BMC)
in finding counter-examples to properties that do not hold.
This is probably due to its backwards DFS search nature,
and leads us to the conclusion of recommending the use of a portfolio approach
combining SMT-based BMC and IC3 when model checking these kinds of
safety instrumented systems.

\newcommand{\IPred}{I}
\newcommand{\Reach}[1]{F_{#1}}
\newcommand{\ReachNext}[1]{F'_{#1}}
\newcommand{\Bound}{k}
\newcommand{\Prop}{P}

%% file: stts.tex
We model timed systems with symbolic timed transition systems
(STTS)~\cite{KindermannJunttilaNiemela:ACSD2011},
a generic formalism allowing modeling of arbitrary control logic structures,
data manipulation, and non-deterministic external inputs.
In a nutshell,
STTSs can be seen as symbolic transition systems
\cite{MannaPnueli:1992} extended with
real-valued clocks and associated constraints.
By using encoding techniques similar to those in
\cite{DBLP:conf/forte/AudemardCKS02,DBLP:journals/entcs/Sorea02},
timed automata (and networks of such) can be efficiently translated into STTS~\cite{KindermannJunttilaNiemela:ACSD2011}.

In the following,
we use standard concepts of propositional and first-order logics.
We assume typed (i.e., sorted) logics,
and
that formulas are interpreted modulo some background theories
(in particular, linear arithmetics over reals);
see e.g.~\cite{Handbook:SMT} and references therein.
%
%
%
If $Y = \{y_1,...,y_l\}$ is a set of variables
and $\phi$ 
formula over $Y$,
then
$Y' = \{y'_1,...,y'_l\}$ is the set of corresponding similarly
typed \emph{next-state variables}
and
$\phi'$ is obtained from $\phi$ by
replacing each variable $y_j$ with $y'_j$.
%
%
Similarly,
if $\psi$ is a formula over $Y \cup Y'$,
then, for each $\Index \in \Naturals$, the formula $\AtStep{\Index}{\psi}$
is obtained by replacing $y_j$ with $\AtStep{\Index}{y_j}$
and
$y'_j$ with $\AtStep{\Index+1}{y_j}$, of the same types.
For example, if $\psi = {(c'_2 \le c_1 + \delta) \land x'_1}$,
then $\AtStep{4}{\psi} = {(\AtStep{5}{c_2} \le \AtStep{4}{c_1} + \AtStep{4}{\delta}) \land \AtStep{5}{x_1}}$.
%
%

An STTS (or simply a system) is a tuple
$\Tuple{\SVars,\Clocks,\Init,\Invar,\Trans,\Resets}$,
where
\begin{itemize}
\item
  $\SVars = \Set{\SVar_1,...,\SVar_n}$ is a finite set of finite domain
  \emph{state variables},
\item
  $\Clocks = \Set{\Clock_1,...,\Clock_m}$ is a finite set of real-valued
  clock variables (or simply \emph{clocks}),
\item
  $\Init$ is a formula over $\SVars$ describing
  the initial states of the system,
\item
  $\Invar$ is a formula over $\SVars \cup \Clocks$ specifying
  a state invariant
  (throughout the paper, we assume the state invariants to be convex, as defined later),
\item
  $\Trans$ is the \emph{transition relation formula} over
  $\SVars \cup \Clocks \cup \SVarsNext$, and
\item
  $\Resets$ associates each clock $\Clock \in \Clocks$
  a  \emph{reset condition} formula $\Reset{\Clock}$
  over $\SVars \cup \Clocks \cup \SVarsNext$.
\end{itemize}
Like in timed automata context,
we require that in all the formulas 
in the system
the use of clock variables is restricted 
to
atoms 
of the form $\Clock \bowtie \Constant$, 
where $\Clock \in \Clocks$ is a clock variable, ${\bowtie} \in \{<,\le,=,\ge,>\}$ and $\Constant \in \Integers$.
Observe that,
as in the timed automata context as well,
one could use rational constants in systems and
then scale them to integers in a behavior and property preserving way.
A system is \emph{untimed} if it does not have any clock variables.
For the sake of readability only, we do not consider the so-called
urgency constraints~\cite{KindermannJunttilaNiemela:ACSD2011} in this paper.

The semantics of an STTS is defined by its states and
how they may evolve to others.
A \emph{state} is simply an interpretation over $\SVars \cup \Clocks$.
A state $\Cube$ is \emph{valid} if it respects the state invariant,
i.e.~$\Models{\Invar}{\Cube}$. 
A state $\Cube$ is an \emph{initial state} if it is valid,
$\Models{\Init}{\Cube}$,
and
$\Cube(\Clock)=0$ for each clock $\Clock \in \Clocks$.
Given a state $\Cube$ and $\Delay \in \RealsNonNeg$,
we denote by $\Cube + \Delay$ the state where clocks have increased by $\Delay$,
i.e.~%
$(\Cube+\Delay)(\Clock) = \Cube(\Clock)+\Delay$ for each clock $\Clock \in \Clocks$
and
$(\Cube+\Delay)(\SVar) = \Cube(\SVar)$ when $\SVar \in \SVars$.
A valid state $\Cube$ may evolve into a successor state $\Cubep$,
denoted by $\Cube \Succ \Cubep$,
if $\Cubep$ is also valid and either of the following holds:
\begin{enumerate}
\item
  \emph{Discrete step}:
  (i) the current and next state interpretations
  evaluate the transition relation to true,
  i.e.\ $\Models{\Trans}{\gamma}$
  where $\gamma(y) = \Cube(y)$ when $y \in {\SVars \cup \Clocks}$
  and $\gamma(\SVarNext) = \Cubep(\SVar)$ when $\SVarNext \in \SVarsNext$,
  and
  (ii)
  each clock either resets or keeps its value:
  for each clock $\Clock \in \Clocks$,
  $\Cubep(\Clock) = 0$ if $\Models{\Reset{\Clock}}{\gamma}$
  and $\Cubep(\Clock) = \Cube(\Clock)$ otherwise.
\item
  \emph{Time elapse step}:
  (i) some amount of time elapses:
  $\Cubep = \Cube+\Delay$ for some $\Delay \in \RealsNonNeg$,
  and
  (ii) the state invariant is respected in the states in between:
  $\Cube + \mu$ is valid for all $0 < \mu \le \Delay$.
\end{enumerate}
A path is a finite sequence $\Cube_0 \Cube_1 ... \Cube_l$
of states such that $\Cube_i \Succ \Cube_{i+1}$ holds for each consecutive
pair of states in the path.
%
%
A state is reachable if there is a path from an initial state to that state.
A \emph{property} $\Property$ is a formula over the state variables
$\SVars$ and the clock variables $\Clocks$, adhering to the same restrictions
on the use of clock variables as the system's formulas.
In this paper we are interested in solving the problem
whether the given state property $\Property$ is an invariant,
i.e.~whether $\Property$ holds in all the reachable states of the system.

As in the contexts of timed automata and linear hybrid automata,
we require the state invariants in STTSs to be convex.
Formally, a state invariant is \emph{convex}
if for all states $\Cube$ and for all $0 \le \DelayTwo \le \Delay$
it holds that whenever
$\Models{\Invar}{\Cube}$ and $\Models{\Invar}{(\Cube+\Delay)}$,
also $\Models{\Invar}{(\Cube+\DelayTwo)}$.
%
Thus, a state invariant cannot become false and then true again during
a time-elapse step, making condition (ii) of time-elapse steps to always hold.
Convexity is easy to test with one call to an SMT solver.


\begin{example}\label{ex:stts}
  As an example,
  consider an STTS modeling the timer $\XRTimerb$ in
  the safety instrumented system in Fig.~\ref{f:sis}(a).
  The STTS has the clock variable $\XTimerb$ which is reset
  when a discrete step makes the signal $\XInput$ true,
  i.e.~$\Reset{\XTimerb} = (\neg\XInput \land \XInput')$,
  corresponding to the activation of the timer.
  The output signal $\XOutput$ is initially false,
  i.e.~$\Init$ contains the conjunct $(\neg\XOutput)$.
  It changes to true when the signal $\XInput$ does
  and
  then stays true for two seconds.
  These properties are captured by the conjunct
  $(\XOutput' \Iff (\neg\XInput \land \XInput')\lor(\XOutput\land(\XTimerb<2)))$
  in the transition relation $\Trans$.
  To force the timer output to be reset after two seconds,
  $\Inv$ contains the conjunct
  $(\XOutput \Implies (\XTimerb \le 2))$.
  %
\end{example}

\begin{figure}[t]
  \centering
  \begin{tabular}{l@{\qquad\qquad}l}
    \includegraphics[width=.5\textwidth]{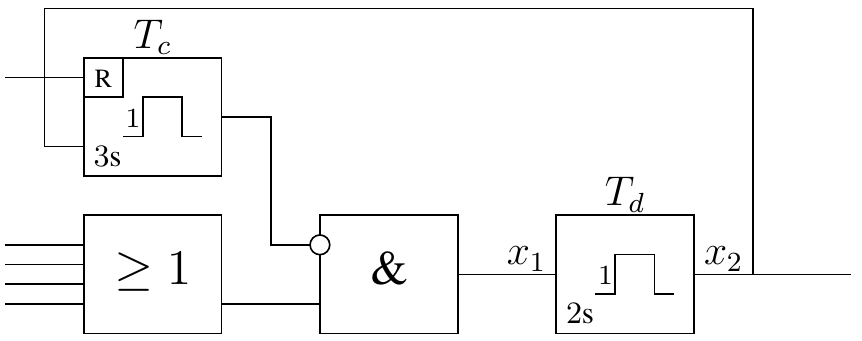}
    &
    \includegraphics[width=.25\textwidth]{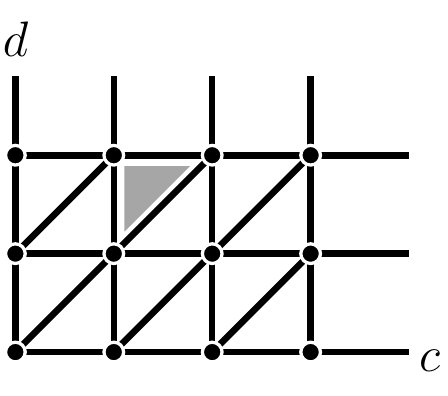}
    \\
    (a) A part of a timed safety instrumented system &
    (b) Clock regions
  \end{tabular}
  \caption{Illustrations of safety instrumented systems and regions.}
  \label{f:sis}
  \label{f:regions}
\end{figure}

\paragraph{Regions.}
A conceptual tool for handling the infinite state space of an STTS
is the region abstraction~\cite{DBLP:journals/tcs/AlurD94}.
%
%
For a non-negative real number $a \in \RealsNonNeg$,
let $\CFrac{a}$ be its fractional part,
i.e.~$a = \Floor{a} + \CFrac{a}$ and $0 \le \CFrac{a} < 1$.
%
%
Let $\ClockValuation$ be an interpretation over $\Clocks$
(also called a \emph{clock valuation}).
%
%
Furthermore,
let $\ClockMax{\Clock}$ be the maximum (relevant) value of the clock $\Clock$,
i.e.\ the largest constant that $\Clock$ is compared to in $\Invar$, $\Trans$, $\Resets$ or $\Property$.
Two clock valuations $\ClockValuation$ and $\AnotherClockValuation$
belong to the same equivalence class called \emph{region},
denoted by $\ClockValuation \SameRegion \AnotherClockValuation$,
if for all clocks $\Clock, \AnotherClock \in \Clocks$
\begin{enumerate}
\item
  either
  (i) $\CInt{\ClockValuation(\Clock)} = \CInt{\AnotherClockValuation(\Clock)}$
  or
  (ii)
  $\ClockValuation(\Clock) > \ClockMax{\Clock}$ and
  $\AnotherClockValuation(\Clock) > \ClockMax{\Clock}$;
\item
  if $\ClockValuation(\Clock) \leq \ClockMax{\Clock}$,
  then
  $\CFrac{\ClockValuation(\Clock)} = 0$ iff
  $\CFrac{\AnotherClockValuation(\Clock)} = 0$;
   and
\item
  if
  $\ClockValuation(\Clock) \leq \ClockMax{\Clock}$ and
  $\ClockValuation(\AnotherClock) \leq \ClockMax{\AnotherClock}$,
  then 
  $\CFrac{\ClockValuation(\Clock)} \leq \CFrac{\ClockValuation(\AnotherClock)}$
  iff
  $\CFrac{\AnotherClockValuation(\Clock)} \leq \CFrac{\AnotherClockValuation(\AnotherClock)}$.
\end{enumerate}
Figure~\ref{f:regions}(b) illustrates the region abstraction for
an STTS with two clocks, $\XTimera$ and $\XTimerb$,
with $\ClockMax{\XTimera} = 3$ and $\ClockMax{\XTimerb} = 2$.
The thick black lines, thick black dots and the areas in between
the thick black lines each represent a different region.
The region in which
$\CInt{\ClockValuation(\XTimera)} = \CInt{\ClockValuation(\XTimerb)} = 1$
and
$0 < \CFrac{\ClockValuation(\XTimera)} < \CFrac{\ClockValuation(\XTimerb)}$
is highlighted in gray.

Two states, $\Cube$ and $\Cubep$,
are in the same region,
denoted by $\Cube \SameRegion \Cubep$,
if they agree on the values of the state variables and
are in the same region when restricted to clock variables.

Due to the restrictions imposed on the use of clock variables,
states in the same region are
(i) indistinguishable for predicates, meaning that
$\Models{\Init}{\Cubep}$ iff $\Models{\Init}{\Cube}$,
$\Models{\Invar}{\Cubep}$ iff $\Models{\Invar}{\Cube}$, and
$\Models{\Property}{\Cubep}$ iff $\Models{\Property}{\Cube}$
whenever $\Cube \SameRegion \Cubep$, and
(ii)
forward bisimilar:
if $\Cube \Succ \Cube'$ and $\Cube \SameRegion \Cubep$,
then there exists a $\Cubep'$ such that
$\Cubep \Succ \Cubep'$ and $\Cube' \SameRegion \Cubep'$.


\paragraph{Formula Representation with Combined Steps.}
To simplify the exposition, to reduce amount of redundancy in paths,
and to enable some optimizations,
we introduce a \emph{formula representation} for STTSs
that exploits a well-known observation:
for reachability checking,
it is enough to consider paths where discrete steps and
time elapse steps alternate, as
two consecutive time elapse steps can be merged into one and zero duration time elapse steps can be added in between discrete steps.
For a  given STTS $\STTS$,
we define the following formulas:
\begin{itemize}
\item
 $\Einvar \Def {\Invar \wedge \ForAllClocks \Clock \geq 0}$.
 Now $\Models{\Einvar}{\Cube}$ for a state $\Cube$ iff
 $\Cube$ is a valid state and all clock values are non-negative.
\item
 $\Einit \Def \Init \wedge \bigwedge_{\Clock\in\Clocks} {\Clock = \FixedClock}$
 for a free real-valued variable $\FixedClock$.
 Now $\Models{\Einit}{\Cube}$ iff,
 forgetting the state validity requirements,
 $\Cube$ is a state reachable from an initial state with
 time elapse steps only.
\item
 $\Etrans \Def
  \Trans \land {\Deltavar \geq 0} \land
  {\ForAllClocks (\Reset{\Clock} \Implies {\Next{\Clock} = \Deltavar})} \land
  {\ForAllClocks ({\neg \Reset{\Clock}} \Implies {\Next{\Clock} = \Clock + \Deltavar})}$.
 Thus, a state $\Cubep$ is reachable from a state $\Cube$
 with one discrete step followed by one time elapse step
 iff
 $\Models{\Etrans}{\pi}$ for the valuation $\pi$ on
 $\SVars \cup \Clocks \cup \SVarsNext \cup \ClocksNext$
 mapping each $z \in {\SVars \cup \Clocks}$ to $\Cube(z)$
 and
 each $\Next{z} \in {\SVarsNext \cup \ClocksNext}$ to $\Cubep(z)$.
\end{itemize}

%% file: k-ind.tex
The $\K$-induction method \cite{SheeranEtAl:FMCAD2000,EenSorensson:2003}
inductively proves a reachability property for a system or
discovers a counter-example while trying to prove the property.
In the following, we will extend $\K$-induction,
which was originally proposed as a verification method for finite-state systems,
to a complete verification method for STTS.

As the base case of an inductive proof, $\K$-induction shows that no bad
state can be reached within $\K$ steps starting from an initial state
for some $\K \in \Naturals$. As the inductive step, $\K$-induction shows that it is impossible under the transition relation of the system to have a path consisting for $\K$ good (property-satisfying) states followed by a bad (property-violating) state. Together, base case and inductive step prove that the property holds in any reachable state.

For an untimed system $\UntimedSys$,
both the base case and inductive step can be proven using a SAT solver.
The base case holds iff the formula
$\AtStep{0}{\Utinit}
\land
\bigwedge_{\Index = 0}^{\K} \AtStep{\Index}{\Invar}
\land
\bigwedge_{\Index = 0}^{\K-1}\AtStep{\Index}{\Trans}
\land
\bigwedge_{\Index = 0}^{\K-1}\AtStep{\Index}{\Property}
\land
\neg\AtStep{\K}{\Property}$
is unsatisfiable. Likewise, the inductive step holds iff the formula
$\bigwedge_{\Index = 0}^{\K} \AtStep{\Index}{\Invar}
\wedge \bigwedge_{\Index = 0}^{\K-1}\AtStep{\Index}{\Trans}
\wedge \bigwedge_{\Index = 0}^{\K-1}\AtStep{\Index}{\Property}
\wedge \neg\AtStep{\K}{\Property}$
is unsatisfiable.
Initially, $\K$-induction attempts an inductive proof with $\K=0$. If unsuccessful, $\K$ is increased until the inductive proof succeeds or a counter-example is found while checking the base case.
Note that the large overlap both between the formulas for checking base case and inductive step and between the checks before and after increasing $\K$ can be exploited by incremental SAT solvers \cite{EenSorensson:2003}.

While correct, the described approach is not complete due to the fact that the induction step is not guaranteed to hold even if the property checked is satisfied by the system. $\K$-induction can, however, be made complete for finite-state systems by only considering simple (non-looping) paths when checking the inductive step.
The most straightforward way to enforce paths to be simple, is to add a quadratic number disequality constraints to the SAT formula, requiring any pair of states to be distinct. Experimental evidence, however, suggests that it is beneficial to only add disequality constraints for pairs of states for which it is observed that disequality constraints are needed \cite{EenSorensson:2003}. 
%

\paragraph{$\K$-induction for STTS.}
Both base case and inductive step formulas can be applied to an STTS $\STTS$ simply by replacing $\Utinit$, $\Uttrans$ and $\Utinvar$ in these formulas by $\Einit$, $\Etrans$ and $\Einvar$ and using an SMT solver instead of a SAT solver.
However, unlike for untimed systems,
termination is not even guaranteed when adding disequality constraints.
For untimed systems,
disequality constraints guarantee termination due to the fact that
in a finite state system, there are no simple paths of infinite length 
and, thus, the simple path inductive step check is guaranteed to be unsatisfiable with sufficiently large~$\K$.
Timed systems, in contrast, typically have 
no upper bound for the length of a simple path and, thus, disequality constraints are not sufficient for completeness.
%
However, the infinite state space of an STTS can be split into a finite number of regions.
%
Thus, any reasoning made for finite state systems can be applied to regions of states.
%
In particular, $\K$-induction is complete and correct
when only paths that do not visit two states belonging to the same region are considered in the inductive step \cite{deMouraEtAl:CAV2003}.
By enforcing this property on inductive step paths using region-disequality
constraints,
complete $\K$-induction can be performed using $\Einit$, $\Etrans$ and $\Einvar$ (almost) without modification.

In order to specify that two states of an STTS belong to different regions,
region-disequality constraints need to individually constrain the integer and
fractional parts of clock values.
As only some SMT-solvers, such as Yices \cite{DBLP:conf/cav/DutertreM06},
allow referring to integer and fractional parts of real-valued variables,
we provide a region-disequality constraint encoding that does not rely on
such a feature.\footnote{In \cite{KindermannJunttilaNiemela:submitted} we give an alternative encoding for region-disequality constraints in a BMC setting.}
Instead, we split each clock variable $\Clock$ into two variables:
$\Int{\Clock}$ represents the integer
and
$\Frac{\Clock}$ the fractional part of $\Clock$'s value.
This ``splitting of clocks'' requires rewriting of
$\Einit$, $\Etrans$ and $\Einvar$ by
replacing each atom involving a clock with
a formula as follows:

\begin{center}
 \footnotesize
   \begin{tabularx}{\textwidth}{| l | X |c| l | X |}
    \cline{1-2} \cline{4-5}
    Atom & Replacement, $n \in \Naturals$ &&
    Atom & Replacement, $n \in \Naturals$
    \\
    \cline{1-2} \cline{4-5}
    \cline{1-2} \cline{4-5}
    $\Clock < \Constant$ & $\Int{\Clock} < \Constant$
    &&
    $\Clock \leq \Constant$ & $\Int{\Clock} < \Constant \vee (\Int{\Clock} = \Constant \wedge \Frac{\Clock} = 0)$
    \\
    \cline{1-2} \cline{4-5}
    $\Clock > \Constant$ & $\Int{\Clock} > \Constant \vee (\Int{\Clock} = \Constant \wedge \Frac{\Clock} > 0)$
    &&
    $\Clock \geq \Constant$ & $\Int{\Clock} > \Constant$
    \\
    \cline{1-2} \cline{4-5}
    $\Clock = \Constant$ & $\Int{\Clock} = \Constant \wedge \Frac{\Clock} = 0$
    &&
    $\Clock = \FixedClock$ & $\Int{\Clock} = \Int{\FixedClock} \wedge \Frac{\Clock} = \Frac{\FixedClock}$
    \\
    \cline{1-2} \cline{4-5}
    $\Next{\Clock} = \Deltavar$ & $\Int{\Clock} = \Int{\Deltavar} \wedge \Frac{\Clock} = \Frac{\Deltavar}$ 
    &
    \multicolumn{3}{ c }{}
    \\ 
    \hline
    $\Next{\Clock} = \Clock + \Deltavar$ & \multicolumn{4}{ p{9.5cm} |}{$%
      ((\Frac{\Clock} + \Frac{\Deltavar} < 1) \Implies  %
      (\Next{\Int{\Clock}} = \Int{\Clock} + \Int{\Deltavar} \wedge \Next{\Frac{\Clock}} = \Frac{\Clock} + \Frac{\Deltavar} )) %
      \wedge %
      ((\Frac{\Clock} + \Frac{\Deltavar} \geq 1) \Implies %
      (\Next{\Int{\Clock}} = \Int{\Clock} + \Int{\Deltavar} + 1 \wedge \Next{\Frac{\Clock}} = \Frac{\Clock} + \Frac{\Deltavar} - 1)) %
    $} \\
    \hline
  \end{tabularx}
\end{center}

\newcommand{\DiffRegions}[2]{\mathit{DiffRegion}^{[#1,#2]}}
\newcommand{\DiffRegionsPrec}[2]{\mathit{DiffRegion}_{\precsim}^{[#1,#2]}}
Then,
two states with indices $\Index$ and $\AnotherIndex$ can be forced to
be in different regions by the following
region-disequality constraint $\DiffRegions{\Index}{\AnotherIndex}$:
\begin{small}
\begin{eqnarray*}
\bigvee_{\SVar \in \SVars} \AtStep{\Index}{\SVar} \neq \AtStep{\AnotherIndex}{\SVar}
& \lor &
\ForAnyClock ({\AtStep{\Index}{\Int{\Clock}} \neq \AtStep{\AnotherIndex}{\Int{\Clock}}} \land (\neg\AtStep{\Index}{\AtMax{\Clock}} \lor \neg\AtStep{\AnotherIndex}{\AtMax{\Clock}}))
\\
& \lor &
\ForAnyClock
({\neg\AtStep{\Index}{\AtMax{\Clock}}} \land
 \neg(\AtStep{\Index}{\Frac{\Clock}} = 0 \Iff
      \AtStep{\AnotherIndex}{\Frac{\Clock}} = 0)
)
\\
& \lor &
\ForAnyClockPair(
        {\neg\AtStep{\Index}{\AtMax{\Clock}}}
  \land {\neg\AtStep{\Index}{\AtMax{\AnotherClock}}}
  \land {\neg((\AtStep{\Index}{\Frac{\Clock}} \le \AtStep{\Index}{\Frac{\AnotherClock}})
	      \Iff
              (\AtStep{\AnotherIndex}{\Frac{\Clock}} \le \AtStep{\AnotherIndex}{\Frac{\AnotherClock}}))}
)
\end{eqnarray*}
\end{small}%
%
where the shorthand
$\AtStep{\Index}{\AtMax{\Clock}} \Def
 {{\AtStep{\Index}{\Int{\Clock}} > \ClockMax{\Clock}} \lor (\AtStep{\Index}{\Int{\Clock}} = \ClockMax{\Clock} \land \Frac{\Clock} > 0)}$
detects whether the clock $\Clock$ exceeds its maximum relevant value $\ClockMax{\Clock}$.

%% file: tic3.tex
In this section,
we first describe the IC3 algorithm~\cite{DBLP:conf/vmcai/Bradley11}
for untimed finite state systems (see also~\cite{een2011efficient}
for an alternative, complementary account of the algorithm).
We then show how it can be extended for verifying timed systems
by using region abstraction and SMT solvers.


Like $\K$-induction,
the IC3 algorithm tries to generate an inductive proof for
a given state property $\Property$ on an untimed system
$\System = \UntimedSys$.
But unlike the unrolling-based approach used by $\K$-induction,
proofs generated by the IC3 algorithm only consists of a single
formula $\ICTProof$ satisfying three properties:
(a) $\ICTProof$ is satisfied by any initial state of $\System$,
(b) $\ICTProof$ is satisfied by any successor of any state satisfying $\ICTProof$,
and
(c) $\ICTProof \Implies \Property$.
Properties (a) and (b) serve as base case and inductive step for showing that
the set of states satisfying $\ICTProof$ is an over-approximation of
the states reachable in $\System$ while
property (c) proves that any reachable state satisfies $\Property$.

In order to generate a proof,
the IC3 algorithm builds a sequence of sets of formulas
$\FSet{0}\ldots\FSet{\FDepth}$ satisfying certain properties.
Eventually, one of these sets becomes the proof $\ICTProof$.
Each $\FSetNoIndex$-set represents the set of states satisfying all its formulas.
The properties satisfied by the sequence are
(i) $\Init \wedge \Invar \Implies \FSet{0}$,
(ii) $\FSet{\Index} \Implies \FSet{\Index + 1}$
,
(iii) $\FSet{\Index} \Implies \Property$
, and
(iv) $\FSet{\Index} \wedge \Invar \wedge \Trans \wedge \Next{\Invar} \Implies \Next{\FSet{\Index + 1}}$.
The basic strategy employed by the IC3 algorithm is to add clauses to the $\FSet{\Index}$-sets in a fashion that keeps properties (i) to (iv) intact until $\FSet{\FDepth} \wedge \Invar \wedge \Trans \wedge \Next{\Invar} \Implies \Next{\Property}$.
In this situation, $\FDepth$ can be increase by appending 
$\Set{\Property}$ to the sequence.
The algorithm terminates once $\FSet{\Index} = \FSet{\Index+1}$ for some $\Index$ and provides 
 $\FSet{\Index}$ as a proof.
Upon termination, properties (i) and (ii) imply proof-property (a), property~(iv) and the termination condition $\FSet{\Index} = \FSet{\Index+1}$ imply property (b) and property (iii) implies property (c).
Note that, in practice, property (ii) is enforced by adding any formula added to a given $\FSetNoIndex$-set also to all $\FSetNoIndex$-sets with lower index, i.e. $\FSet{\Index} \subseteq \FSet{\Index-1}$

\begin{figure}[tp]
\begin{algorithmic}[1]
\LOOP
	\IF{$\FSet{\FDepth} \wedge \Invar \wedge \Trans \wedge \Next{\Invar} \wedge \neg \Next{\Property}$ is UNSAT} \label{l:ic3.main.sat}
		\STATE $\FDepth \Assign \FDepth + 1$ \label{l:ic3.main.extend.first}
		\STATE add $\FSet{\FDepth} \Asgn \{\Property\}$ to sequence of $\FSetNoIndex$-sets \label{l:ic3.main.extend.last}
		\STATE propagate() \label{l:ic3.main.propagate}
		\IF{$\FSet{\Index} = \FSet{\Index+1}$ for some $\Index$} \label{l:ic3.main.termination}
			\RETURN $\True$ \COMMENT {Property holds}
		\ENDIF
	\ELSE
		\STATE $\Cube \Asgn$ predecessor of a bad state extracted from the model
		\STATE $success \Asgn \BlockFunc(\Cube)$
		\IF{$\neg success$}
			\RETURN $\False$ \COMMENT {Property violated}
		\ENDIF
	\ENDIF
\ENDLOOP
\end{algorithmic}
\caption{The main loop of IC3}
\label{f:ic3.main}
\end{figure}

After sketching the basic strategy, we will now take a closer look at the algorithm.
Note, however, that the description given is only a simplified version of the algorithm that focuses on the aspects that are relevant with respect to extending it for STTS.
Figure~\ref{f:ic3.main} shows the main loop of the IC3 algorithm. In each iteration, the algorithm first checks whether or not it is currently possible to extend the sequence of $\FSetNoIndex$-sets by appending $\Property$. Note that as appending $\Property$ will never result in properties (i) to (iii) being violated, it is sufficient to check whether extending the 
sequence would violate property~(iv). A corresponding SAT call can be found in Line~\ref{l:ic3.main.sat} of Fig.~\ref{f:ic3.main}.
If the SAT call indicates that the sequence can safely extended, the sequence is extended in Lines~\ref{l:ic3.main.extend.first} and~\ref{l:ic3.main.extend.last}.
In the next step,
the $\FSetNoIndex$-set sequence,
clauses may be propagated from $\FSetNoIndex$-sets to subsequent sets in the sequence.
While this step is vital for termination, a more detailed description is omitted here for space limitations.
After propagation, 
the algorithm's termination condition is checked in Line~\ref{l:ic3.main.termination}.

Of course, the SAT check in Line~\ref{l:ic3.main.sat} may as well indicate that the $\FSetNoIndex$-sequence may currently not be extended without violating property~(iv).
In this case, a state $\Cube$ that satisfies $\FSet{\FDepth}$ and has a bad successor can be extracted from the model returned by the SAT solver.
As $\Cube$ prevents the sequence from being extended, the algorithm attempts to drop $\Cube$ from (the set of states represented by) $\FSet{\FDepth}$ by
adding a clause that implies $\neg\Cube$\footnote{In a slight abuse of notation, we interpret a state $\Cube$ as formula $\bigwedge_{\ClockOrSVal \in \Clocks \cup \SVars} \ClockOrSVal = \Eval{\ClockOrSVal}{\Cube}$ where appropriate.}.
The corresponding subroutine call,
$\BlockFunc(\Cube)$,
may also need to add further clauses also to other $\FSetNoIndex$-sets
than $\FSet{\FDepth}$ in order to ensure that properties of
the sequence remain satisfied.

\begin{figure}[tp]
\begin{algorithmic}[1]
\STATE $\Queue \Asgn$ priority queue containing $\Tuple{\Cube, \FDepth}$
\WHILE{$\Queue$ not empty}
	\STATE $\Cube, \Index \Asgn \Queue.popMin()$
	\IF{$\Index = 0$} \label{ic3.block.counter-example} \label{l:ic3.block.counter-example}
		\RETURN $\False$ \COMMENT{Counter-example found}
	\ENDIF
	\IF{$\FSet{\Index-1}\wedge\neg\Cube\wedge\Trans\wedge\Invar\wedge\Next{\Invar}\wedge\Next{\Cube}$ is SAT} \label{l:ic3.block.sat}
		\STATE $\AnotherCube \Asgn$ predecessor of $\Cube$ extracted from the model
		\STATE $\Queue.add(\Tuple{\Cube, \Index})$ \label{l:ic3.block.readd}
		\STATE $\Queue.add(\Tuple{\AnotherCube, \Index - 1})$ \label{l:ic3.block.addnew}
	\ELSE
		\STATE $\FSet{\Index}.add(generalize(\neg \Cube))$ \label{l:ic3.block.addclause}
		\IF{$\Index < \FDepth$}
			\STATE $\Queue.add(\Tuple{\Cube, \Index + 1})$
		\ENDIF
	\ENDIF
\ENDWHILE
\RETURN $\True$
\end{algorithmic}
\caption{The $\BlockFunc(\Cube)$ sub-routine}
\label{f:ic3.block}
\end{figure}

The $\BlockFunc(\Cube)$ subroutine, outlined in Figure~\ref{f:ic3.block}, operates on a list of proof obligation, each being a pair of a state and an index.
An obligation $\Tuple{\Cube, \Index}$ indicates that it is necessary to drop $\Cube$ from $\FSet{\Index}$ before the main loop of the algorithm can continue.
Initially, the only proof obligation is to drop the state provided as an argument from $\FSet{\FDepth}$.
For any proof obligation $\Tuple{\Cube, \Index}$, the $\BlockFunc$ subroutine in Line~\ref{l:ic3.block.sat} checks whether or not $\Cube$ has a predecessor $\AnotherCube$ in $\FSet{\Index - 1}$.
Such a predecessor prevents $\Cube$ from being excluded from $\FSet{\Index}$ without violating property~(iv).
Thus, if 
a predecessor is found, the obligation $\Tuple{\Cube, \Index}$ can not be fulfilled immediately and is added to the set of open obligations again in Line~\ref{l:ic3.block.readd}.
Furthermore, $\AnotherCube$ has to be excluded from $\FSet{\Index - 1}$ before $\Cube$ can be excluded from $\FSet{\Index}$.
This is reflected by the obligation $\Tuple{\AnotherCube, \Index-1}$ also being added to the set of open obligations in Line~\ref{l:ic3.block.addnew}.

If the SAT call in Line~\ref{l:ic3.block.sat} is unsatisfiable, then $\Cube$ has no predecessor in $\FSet{\Index-1}$ and can safely be excluded from $\FSet{\Index}$ without violating property~(iv). The state $\Cube$ is excluded by adding a generalization of the clause $\neg \Cube$ to $\FSet{\Index}$. More precisely, the algorithm attempts to drop literals from 
$\neg \Cube$ in a way that preserves properties (i) and (iv) before adding the resulting clause to $\FSet{\Index}$. Without this generalization step, states would be excluded one at a time from the $\FSetNoIndex$-sets resulting in a method akin to explicit state model checking.

So far, it has been assumed that $\Property$ holds. If this is not the case, then the main loop will eventually pass a predecessor of a bad state reachable in $\System$ to $\BlockFunc$. In such a situation,  $\BlockFunc$  essentially performs a backwards depth-first search that eventually leads to an initial state of $\System$, which is detected in Line~\ref{l:ic3.block.counter-example}. Note that it is straightforward to extract a counter-example from the proof obligations if it is detected that $\Property$ does not hold.

Note that, while sufficient for explain our extensions, only a simplified version of the IC3 algorithm has been described. Most notably, the complete version of the algorithm additionally aims to satisfy proof obligations for multiple successive $\FSetNoIndex$-sets at a time if possible and performs generalization based on unsatisfiable cores obtained from SAT calls in various locations. For a description of these techniques as well as complete arguments for correctness and completeness of the approach refer to \cite{DBLP:conf/vmcai/Bradley11,een2011efficient}.



\newcommand{\RCube}{\tilde{s}}
\newcommand{\TRCube}{\tilde{s}_{\precsim}}

\paragraph{Extending IC3 for timed systems.}
As was the case with $\K$-induction,
the key to extending the IC3 algorithm to timed systems is
the region abstraction.
Again,
we will use an SMT-solver instead of a SAT-solver
and the combined step encoding $\Einit$, $\Einvar$, and $\Etrans$ on an STTS will replace $\Init$, $\Invar$ and $\Trans$.
%
%
To operate on the region level,
we lift each concrete state in a satisfying interpretation
returned by the SMT solver
into to the region level in the IC3 algorithm code
whenever it is passed back to the SMT solver again.
%
%
To do this,
given a state $\Cube$,
we construct a conjunction $\RCube$ of atoms
such that $\RCube$ represents all the states in the same region
as $\Cube$,
i.e.~for any state $\Cubep$ it holds
$\Models{\RCube}{\Cubep}$ iff $\Cubep \SameRegion \Cube$.
Formally,
$\RCube$ is the conjunction of the atoms given by the following rules:
\begin{enumerate}
\item
  For each state variable $\SVar \in \SVars$,
  add the atom $(\SVar = \Cube(\SVar))$.
\item
  For each clock $\Clock$ with
  $\ClockValue{\Clock} > \ClockMax{\Clock}$,
  add the atom $(\Clock > \ClockMax{\Clock})$.
\item
  For each clock $\Clock$ with
  $\ClockValue{\Clock} \le \ClockMax{\Clock}$
  and
  $\FracValue{\Clock} = 0$,
  add the atoms
  $(\Clock \le \ClockValue{\Clock})$ and
  $(\Clock \ge \ClockValue{\Clock})$.
  Two atoms are added instead of $(\Clock = \ClockValue{\Clock})$
  so that the clause generalization sub-routine has more possibilities
  for relaxing $\neg\RCube$.
\item
  For each clock $\Clock$ with
  $\ClockValue{\Clock} < \ClockMax{\Clock}$
  and
  $\FracValue{\Clock} \neq 0$,
  add the atoms
  $(\Clock > \Floor{\ClockValue{\Clock}})$ and
  $(\Clock < \Ceil{\ClockValue{\Clock}})$.
\item
  For each pair $\Clock,\AnotherClock$ of distinct clocks
  with
  $\ClockValue{\Clock} \le \ClockMax{\Clock}$, and
  $\ClockValue{\AnotherClock} \le \ClockMax{\AnotherClock}$,
  \begin{enumerate}
  \item
    if $\FracValue{\Clock} = \FracValue{\AnotherClock}$,
    add the atoms
    $(\AnotherClock \le {\Clock-\IntValue{\Clock}+\IntValue{\AnotherClock}})$
    and
    $(\AnotherClock \ge {\Clock-\IntValue{\Clock}+\IntValue{\AnotherClock}})$,
    and
  \item
    if $\FracValue{\Clock} < \FracValue{\AnotherClock}$,
    add the atom
    $(\AnotherClock > {\Clock-\IntValue{\Clock}+\IntValue{\AnotherClock}})$.
  \end{enumerate}
\end{enumerate}
What is especially convenient here is that,
unlike in the region-disequality constraints in $\K$-induction,
there is no need to directly access the integral and fractional parts
of clock variables in $\RCube$
 because $\RCube$ considers one fixed region.
Indeed, all the atoms concerning clock variables will fall in
the difference logic fragment of linear arithmetics over reals,
having very efficient decision procedures available~\cite{NieuwenhuisOliveras:CAV2005,CottonMaler:SAT2006}.

We now let the IC3 algorithm operate as in the untimed case
except the satisfiability calls are changed to operate on the region level.
Especially,
the formula in Line \ref{l:ic3.block.sat} of Fig.~\ref{f:ic3.block}
is modified to 
$\FSet{\Index-1} \land {\neg\RCube} \land \Trans \land \Invar \land \Next{\Invar} \land \Next{\RCube}$
so that it operates on the region level,
trying to find a predecessor state in the $\FSet{\Index-1}$-set
for \emph{any} state in the same region as $\Cube$.
Furthermore,
in Line~\ref{l:ic3.block.addclause} the clause generalization
is called with the clause $\neg\RCube$ that represents all the states that
are in a different region than $\Cube$;
thus we exclude at least all states in the region of $\Cube$
from $\FSet{\Index}$.
%
%
%
%
No major modifications are required in the clause generalization mechanisms
but they can handle clock atoms in 
the same way the state variable literals are handled.
\begin{example}
  Consider the STTS for the system in Fig.~\ref{f:sis}(a) discussed in
  Ex.~\ref{ex:stts}.
  For the state $\Cube = \Set{\XInput\mapsto\False,\XOutput\mapsto\True,\XTimera \mapsto 1.4,\XTimerb \mapsto 1.65,...}$
  we obtain the conjunction
  $\RCube = ({\neg\XInput} \land {\XOutput} \land (\XTimera>1) \land (\XTimera < 2) \land (\XTimerb>1) \land (\XTimerb < 2) \land (\XTimerb > \XTimera) \land ...)$
  that represents all the states in the region of $\Cube$.
  The clause that excludes the whole region of $\Cube$ is simply
  $(\XInput \lor {\neg\XOutput} \lor (\XTimera \le 1) \lor (\XTimera \ge 2) \lor (\XTimerb \le 1) \lor (\XTimerb \ge 2) \lor (\XTimerb \le \XTimera) \lor ...)$.
\end{example}

The soundness of the timed IC3 algorithm can be argued as follows.
We say that a formula $\phi$ over $\SVars \cup \Clocks$ \emph{respects regions}
if for all states $\Cube$ and $\Cubep$,
$\Cube \SameRegion \Cubep$ implies
that $\Models{\Cube}{\phi}$ iff $\Models{\Cubep}{\phi}$.
By construction, a state property $\Property$
as well as $\RCube$ and $\neg\RCube$ for any state $\Cube$
all respect regions.
Furthermore,
any sub-clause of $\neg\RCube$ returned by the clause generalization
sub-routine also respects regions.
As a result all the clauses in the $\FSetNoIndex$-sets
respect regions and thus the $\FSetNoIndex$-sets
exclude whole regions only.
Furthermore,
the modified formula
$\FSet{\Index-1} \land {\neg\RCube} \land \Trans \land \Invar \land \Next{\Invar} \land \Next{\RCube}$ in Line \ref{l:ic3.block.sat}
of Fig.~\ref{f:ic3.block} is unsatisfiable
iff
no state in the same region as $\Cube$ can be reached from
the $\FSet{\Index-1}$-set;
thus excluding the whole region in Line~\ref{l:ic3.block.addclause}
is correct.

Because the number of regions is finite,
only a finite number of clauses can be added to any $\FSetNoIndex$-set.
As a result,
the argument for termination given in \cite{een2011efficient}
can be applied to the timed IC3 algorithm as well.

%% file: optimizations.tex
\begin{figwindow}[1,r,%
\makebox[45mm][c]{\includegraphics[width=40mm]{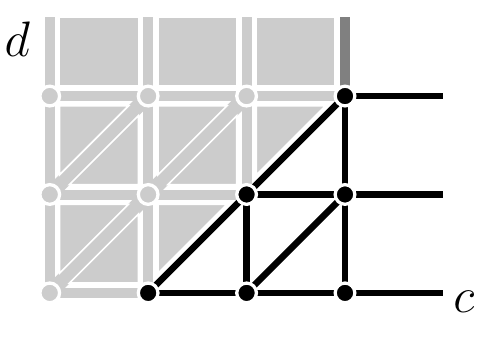}},%
{A time-predecessor clock region
\label{f:optimization}}\vspace{0.1cm}]
We now describe optimizations for timed IC3 and $\K$-induction
that sometimes allow us to exclude more regions at once
during the clause generalization and
in the region-disequality constraints, respectively.
They both exclude time-predecessor regions of a region,
i.e.~regions from whose states one can reach the other region by
just letting time pass.
As an example,
all the light regions and the dark gray region
(with $\XTimera=3$ and $\XTimerb > 2$)
in Fig.~\ref{f:optimization}
are time-predecessors of the dark gray region.
%
%
Formally,
we define that a clock valuation
$\AnotherClockValuation$ is in a time-predecessor region of
the clock valuation $\ClockValuation$,
denoted by $\AnotherClockValuation \TPrec \ClockValuation$,
if for all clocks $\Clock, \AnotherClock \in \Clocks$
all the following hold:
\end{figwindow}
\begin{enumerate}
\item
  Either
  (i) $\ClockValuation(\Clock) > \ClockMax{\Clock}$,
  (ii) 
  $\CFrac{\ClockValuation(\Clock)} = 0$ and
  $\AnotherClockValuation(\Clock) \le \ClockValuation(\Clock)$,
  or
  (iii)
  $\CFrac{\ClockValuation(\Clock)} > 0$ and
  $\AnotherClockValuation(\Clock) < \Ceil{\ClockValuation(\Clock)}$.
\item
  If
  $\ClockValuation(\Clock) \leq \ClockMax{\Clock}$ and
  $\ClockValuation(\AnotherClock) \leq \ClockMax{\AnotherClock}$,
  then
  \begin{enumerate}
  \item[(i)]
    $\CFrac{\ClockValuation(\Clock)} = \CFrac{\ClockValuation(\AnotherClock)}$
    implies
    $\AnotherClockValuation(\AnotherClock) =
     \AnotherClockValuation(\Clock) -
     \Floor{\ClockValuation(\Clock)} +
     \Floor{\ClockValuation(\AnotherClock)}$, and
  \item[(ii)]
    $\CFrac{\ClockValuation(\Clock)} < \CFrac{\ClockValuation(\AnotherClock)}$
    implies
    $\AnotherClockValuation(\AnotherClock) >
     \AnotherClockValuation(\Clock) -
     \Floor{\ClockValuation(\Clock)} +
     \Floor{\ClockValuation(\AnotherClock)}$
    and
    $\AnotherClockValuation(\AnotherClock) <
     \AnotherClockValuation(\Clock) -
     \Floor{\ClockValuation(\Clock)} +
     \Floor{\ClockValuation(\AnotherClock)} +
     1$.
  \end{enumerate}
\item
  If
  $\ClockValuation(\Clock) \leq \ClockMax{\Clock}$ and
  $\ClockValuation(\AnotherClock) > \ClockMax{\AnotherClock}$,
  then
  \begin{enumerate}
  \item[(i)]
    $\CFrac{\ClockValuation(\Clock)} = 0$
    implies
    $\AnotherClockValuation(\AnotherClock) >
     \AnotherClockValuation(\Clock) -
     \Floor{\ClockValuation(\Clock)} +
     \ClockMax{\AnotherClock}$,
    and
  \item[(ii)]
    $\CFrac{\ClockValuation(\Clock)} > 0$
    implies
    $\AnotherClockValuation(\AnotherClock) >
     \AnotherClockValuation(\Clock) -
     \Floor{\ClockValuation(\Clock)} +
     \ClockMax{\AnotherClock} - 1$.
  \end{enumerate}
\end{enumerate}
Observe that $\TPrec$ is a reflexive relation.
A state $\Cubep$ is in a time-predecessor region
of another state $\Cube$,
denoted by $\Cubep \TPrec \Cube$,
if they agree on the values of the state variables and,
when restricted to the clock variables,
$\Cubep$ is in a time-predecessor region of $\Cube$.
%

\subsubsection{Application to IC3.}
\label{ss:opt.ic3}

The timed variant of the IC3 algorithm described in a previous section
excludes an entire region from an $\FSetNoIndex$-set once
a state inside that region (and thus the whole region)
has been found to be unreachable from the previous $\FSetNoIndex$-set.
In this section,
we will argue that it is actually possible to exclude all the time-predecessor
regions at the same time.
By excluding more than one region,
the $\FSetNoIndex$-sets potentially shrink faster
which can lead to improved execution times.
%
%
%
%
This optimization to the IC3 algorithm is based on the following lemma:
\begin{lemma}
  \label{lem:opt}
  Let $\Cube$ be a valid state.
  If none of the states in the region of $\Cube$ can be reached
  from an initial state with one time elapse step
  followed by $\NumSteps$ combined steps,
  then
  none of the valid states in the time-predecessor regions of $\Cube$ can,
  either.
\end{lemma}
\begin{proof}
  Assume that a valid state $\Pred \TPrec \Cube$ is
  reachable in that way.
  Thus,
  (i)
  $\Pred$ satisfies $\Invar$,
  (ii)
  there is a $\Delay \in \RealsNonNeg$ such that
  $\Pred+\Delay \SameRegion \Cube$
  as $\Pred \TPrec \Cube$,
  (iii)
  $\Pred+\Delay$ satisfies $\Invar$ as $\Cube$ does,
  and
  (iv)
  all the states ``in between'' $\Pred$ and $\Cube$
  (i.e.~all the states $\Pred+\Delay'$ with $0 < \Delay' < \Delay$)
  satisfy the convex $\Invar$, too.
  Therefore,
  $\Pred+\Delay$ is reachable
  from an initial state with one time elapse step
  followed by $\NumSteps$ combined steps
  by just ``extending'' the last time elapse step by $\Delay$ units.
  This gives a contradiction as $\Pred+\Delay \SameRegion \Cube$.
  \qed
\end{proof}

Any state considered by the IC3 algorithm is extracted from a model of a SMT
formula containing the system's state invariant as a conjunct and
hence satisfies the invariant.
In addition,
the $\Einit$ and $\Etrans$ formulas used in timed IC3 capture
initial states followed by one time elapse step
and
combined steps, respectively.
Thus,
Lemma~\ref{lem:opt} is applicable to any state found unreachable by
the IC3 algorithm
and
justifies the dropping of all the time-predecessor regions at the same time.
%
%
%
%
%
%
%
%
%
%
Given a state $\Cube$,
we can construct a conjunction $\TRCube$ of atoms
such that $\TRCube$ represents all the states
in the time-predecessor regions of $\Cube$.
Formally,
$\TRCube$ is obtained by instantiating the definition of $\TPrec$
for a concrete state $\Cube$
and
is the conjunction of the atoms given by the following rules:
\begin{enumerate}
\item
  For each state variable $\SVar \in \SVars$,
  add the atom $(\SVar = \Cube(\SVar))$.
\item
  For each clock $\Clock$ with
  $\ClockValue{\Clock} \le \ClockMax{\Clock}$
  and
  $\FracValue{\Clock} = 0$,
  add the atom $(\Clock \le \ClockValue{\Clock})$.
\item
  For each clock $\Clock$ with
  $\ClockValue{\Clock} < \ClockMax{\Clock}$
  and
  $\FracValue{\Clock} \neq 0$,
  add the atom
  $(\Clock < {\lceil\ClockValue{\Clock}\rceil})$.
\item
  For each pair $\Clock,\AnotherClock$ of distinct clocks
  with
  $\ClockValue{\Clock} \le \ClockMax{\Clock}$ and
  $\ClockValue{\AnotherClock} \le \ClockMax{\AnotherClock}$,
  \begin{enumerate}
  \item
    if $\FracValue{\Clock} = \FracValue{\AnotherClock}$,
    add the atoms
    $(\AnotherClock \leq {\Clock-\IntValue{\Clock}+\IntValue{\AnotherClock}})$
    and
    $(\AnotherClock \geq {\Clock-\IntValue{\Clock}+\IntValue{\AnotherClock}})$,
    again using two literals to encode equality for additional clause relaxation possibilities,
    and
  \item
    if $\FracValue{\Clock} < \FracValue{\AnotherClock}$,
    add the atoms
    $(\AnotherClock > {\Clock-\IntValue{\Clock}+\IntValue{\AnotherClock}})$
    and
    $(\AnotherClock < {\Clock-\IntValue{\Clock}+\IntValue{\AnotherClock}+1})$.
  \end{enumerate}
\item
  For each pair $\Clock,\AnotherClock$ of distinct clocks
  with
  $\ClockValue{\Clock} \le \ClockMax{\Clock}$ and
  $\ClockValue{\AnotherClock} > \ClockMax{\AnotherClock}$,
  \begin{enumerate}
  \item
    if $\FracValue{\Clock} = 0$,
    add the atom
    $(\AnotherClock > {\Clock-\IntValue{\Clock}+\ClockMax{\AnotherClock}})$;
    and
  \item
    if $\FracValue{\Clock} > 0$,
    add the atom
    $(\AnotherClock > {\Clock-\IntValue{\Clock}+\ClockMax{\AnotherClock}}-1)$.
  \end{enumerate}
\end{enumerate}
%
Now $\TRCube$ and $\neg \TRCube$ can be used instead of $\Cube$ and $\neg\Cube$ in SMT calls and as argument for clause generalization.
Observe that $\TRCube$ is also in the difference logic fraction of linear arithmetics and does not need to refer to integral or fractional parts of clocks.
\begin{example}
  Consider again the STTS for the system in Fig.~\ref{f:sis}(a) discussed in
  Ex.~\ref{ex:stts}.
  For the state
  $\Cube = \Set{\XInput\mapsto\False,\XOutput\mapsto\True,\XTimera \mapsto 3.0,\XTimerb \mapsto 2.7,...}$
  in the dark gray clock region in Fig.~\ref{f:optimization},
  we get the conjunction
  $\RCube = ({\neg\XInput} \land {\XOutput} \land (\XTimera \le 3) \land (\XTimerb > {\XTimera - 1}) \land ...)$
  representing all the states in the time-predecessor regions.
  %
  %
\end{example}

\subsubsection{Application to $\K$-induction.}

\newcommand{\Cubes}{s^{\textup{d}}}
\newcommand{\Cubeps}{u^{\textup{d}}}

The idea of excluding time-predecessor regions
can also be applied to $\K$-induction.
This is based on the following lemma,
stating that a path of combined steps can be compressed
into a shorter region-equivalent one
if a state in it is in the time-predecessor region of a later state:
\begin{lemma}
  \label{l:k-ind-opt}
  Let
  $\Cube_0
   \Cubes_1 \Cube_1 \ldots
   \Cubes_{i-1} \Cube_{i-1}
   \Cubes_i \Cube_{i} \ldots
   \Cubes_j \Cube_{j} \ldots
   \Cubes_k$
  be a path such that (i) $\Cube_0$ is an initial state,
  (ii) $\Cubes_i \TPrec \Cubes_j$,
  (iii) each step between $\Cube_{l}$ and $\Cubes_{l+1}$
  is a time elapse step,
  and
  (iv) each step between $\Cubes_{l}$ and $\Cube_{l}$
  is a discrete step.
  Then
  $\Cube_0
   \Cubes_1 \Cube_1 \ldots
   \Cubes_{i-1} \Cube_{i-1}
   \Cubeps_j \Cubep_{j} \ldots
   \Cubeps_k$
  with
  $\Cubeps_j \SameRegion \Cubes_j$ for all $j \le l \le k$
  and
  $\Cubep_{j} \SameRegion \Cube_{j}$ for all $j \le l < k$
  is also a path.
\end{lemma}
\begin{proof}
  As $\Cubes_i \TPrec \Cubes_j$ and the state invariants are convex,
  the time elapse step from $\Cube_{i-1}$ to $\Cubes_i$
  can be ``extended'' so that
  a state $\Cubes_i+\Delay \SameRegion \Cubes_j$
  is reached instead.
  Letting $\Cubeps_j$ equal $\Cubes_i+\Delay$,
  the existence of the requested postfix $\Cubeps_j \Cubep_{j} ... \Cubeps_k$
  of the path follows from the forward bisimilarity of
  the states in the same region.
  \qed
\end{proof}
Now this implies that in timed $\K$-induction we can use,
instead of the region-disequality formula $\DiffRegions{i}{j}$,
a stronger formula $\DiffRegionsPrec{i}{j}$ excluding
the state $\AtStep{i}{\Cube}$ from being in a time-predecessor
region of the state $\AtStep{j}{\Cube}$ when $i < j$.
We omit the details but
this formula can be obtained from the definition of the $\TPrec$ relation
in a similar way as the $\DiffRegions{i}{j}$ formula was obtained
from the definition of $\SameRegion$ in Sect.~\ref{s:k-ind}.


%% file: experiments.tex
To determine the usefulness of the described methods, they were evaluated experimentally. Specifically, we were interested in the following questions: how do the methods perform and scale (i) in the area they were designed for, i.e. timed systems with a large amount of non-determinism; (ii) compared to each other; (iii) compared to using discrete time verification methods in a semantics-preserving way; and (iv) outside the area they were designed for, i.e. on models with a low amount of non-determinism?


\paragraph{Setup.} Timed $\K$-induction and the timed IC3 algorithm were implemented in Python, each supporting both region encoding variants.
Using a more efficient programming language like C is likely
to yield only moderate execution time improvements due to a significant fraction of the time being spent by the SMT solver.
As an SMT-solver, Yices~\cite{DBLP:conf/cav/DutertreM06} version 1.0.31 was used.
All experiments were executed on Linux computers with AMD Opteron 2435 CPUs limited to one hour of CPU time and 2 GB of RAM.

\paragraph{Industrial benchmark.}
The first benchmark used is
a model of an emergency diesel generator intended for the use in a nuclear power plant. 
The full model and two sub-models, 
which are sufficient for the verification some of the properties, were used. The numbers of clocks and state variables are 24 and 130 for the full model, 7 and 64 for the
first
 and 6 and 36 for the
 second
 sub-model.
The industrial model has been studied previously and 
found very challenging. Only some partial
results~\cite{Lahtinen-W156} have been obtained using the model checker
NuSMV~\cite{NuSMV} by abstracting model
based on its component structure and then using a discrete time version
of the model. Efforts to verify the abstracted model using the real time
model checker Uppaal~\cite{Uppaaltutorial:2004} were even less 
successful~\cite{Lahtinen-W156}. Likewise, a booleanization-based 
attempt to verify the smallest sub-model was unable to verify all
properties~\cite{KindermannJunttilaNiemela:ACSD2011}.

\begin{table}[tp]\scriptsize
\centering
\caption{Verification times in seconds for industrial benchmarks. Blank cells indicate that the respective property can not be verified on the respective size model.}
\label{t:res-tab}
\input{table}
\end{table}

All four variants of the methods introduced in this paper were applied to the industrial model.
Additionally, the original IC3 implementation \cite{DBLP:conf/vmcai/Bradley11}
in combination with a semantics-preserving booleanization approach~\cite{KindermannJunttilaNiemela:ACSD2011}
was used.
Table~\ref{t:res-tab} shows the resulting execution times.
$\K$-induction did not exceed three seconds for any property.
The timed IC3-approaches performed similarly for most properties but timed out four times. 
Both real-time verification methods performed significantly better than the booleanization~/~IC3 combination, illustrating that development of specialized real time verification methods is worthwhile.

\paragraph{Random properties.}

\newcommand{\cactuswidth}{.43\textwidth}
\begin{figure}[tp]
\centering
\subfigure[Full model, property holds]{\includegraphics[width=\cactuswidth]{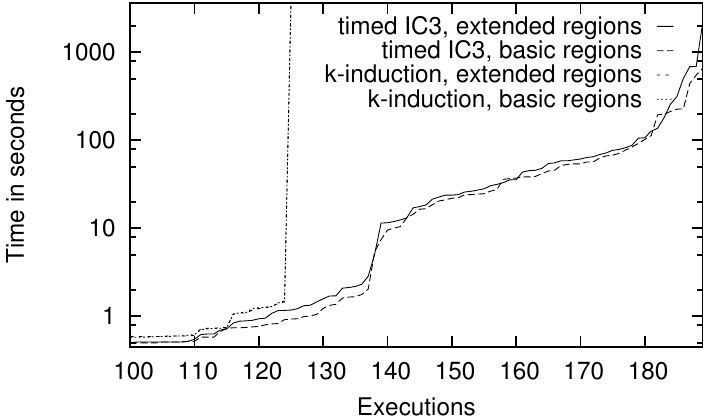}}\qquad
\subfigure[Medium size model, property holds]{\includegraphics[width=\cactuswidth]{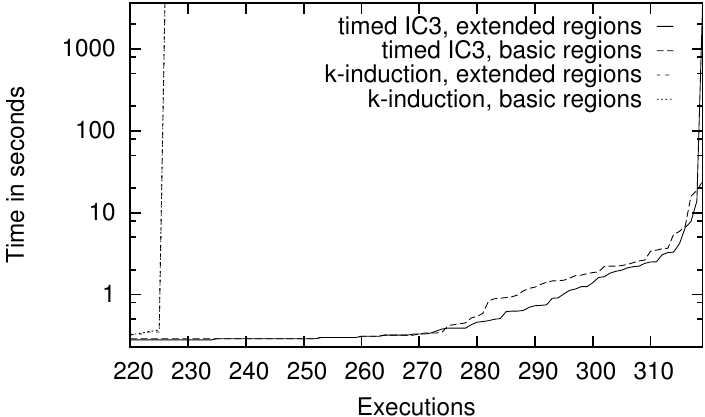}}
\subfigure[Full model, property violated]{\includegraphics[width=\cactuswidth]{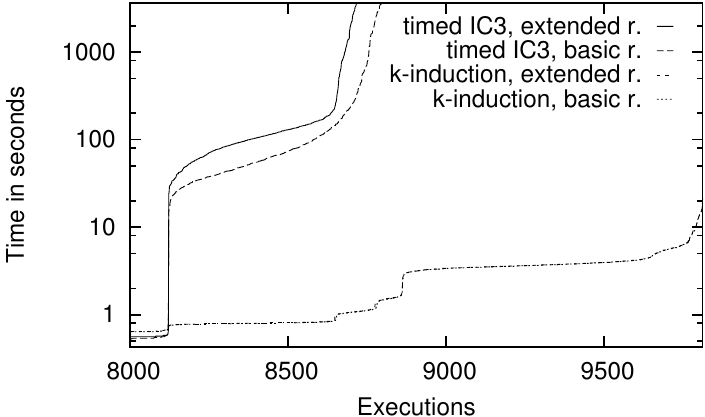}}\qquad
\subfigure[Medium size model, property violated]{\includegraphics[width=\cactuswidth]{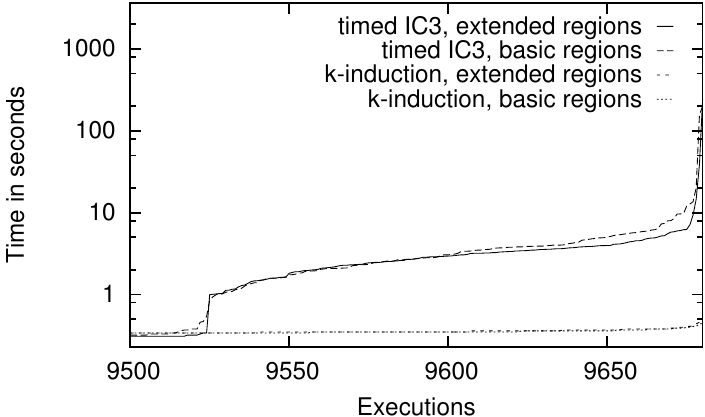}}
\caption{Time required to verify by numbers of properties for randomly generated properties}
\label{f:cactus}
\end{figure}

While the industrial benchmark showed 
that the methods work well in the area they were designed for, 
execution times were generally too low to compare the different methods and variants.
Therefore,
10000 additional random properties were generated each for the full model and the medium size sub-model, each property being a three literal clause using state variables and / or clocks. 
Figure~\ref{f:cactus} shows the resulting execution times. Note that all methods timed out for one property on the medium size model, which then could not be considered in the plots due to not being known whether it holds.
%
For violated random properties, $k$-induction performed very well, 
due to 
its bounded model checking component.
For properties that hold, in contrast, timed IC3 performed significantly better. Executing both methods (or timed IC3 and bounded model checking) in parallel could combine their strengths. 

Using time-predecessor regions made no difference for $k$-induction.
For the timed IC3 algorithm, 
their effect
depended on the size of the model used. 
A performance increase was observed for the medium size model, contrasting a performance decrease for the large model.
A likely explanation for this behavior is the large number of clocks used in the large model.
While the 
time-predecessor
region encoding uses fewer literals referring to a single clock than the original region encoding, it contains more literals comparing two clocks. 
Thus, the size of clauses grows quicker in the number of clocks used for time-predecessor regions, which eventually outweighs the gain of excluding more states at once.

\paragraph{Fischer protocol.}

\begin{figure}[tp]
\centering
\includegraphics[width=.9\textwidth]{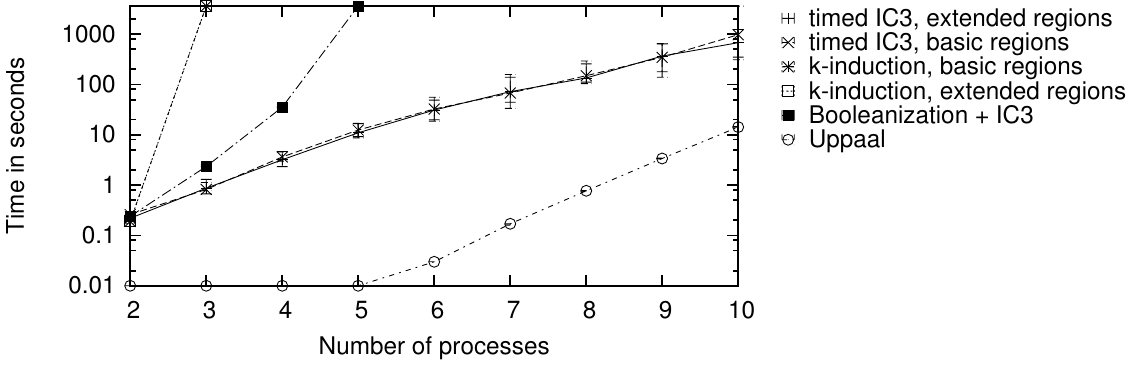}
\caption{Verification time the Fischer protocol
(min, max and median of 11 executions)}
\label{f:res.fischer}
\end{figure}

As a third benchmark, the Fischer mutual exclusion protocol, a standard benchmark for timed verification, was used.
In addition to 
the five methods used for the industrial method, Uppaal~\cite{Uppaaltutorial:2004} version 4.0.11, a model checker for networks of timed automata,  was used. 
Unlike the industrial benchmark, the Fischer protocol is fairly
deterministic 
and, thus,
could be expected to favor Uppaal over the fully-symbolic methods. 
%
%
%
Figure~\ref{f:res.fischer} shows the execution times for verifying the Fischer property with a varying number of processes.
While timed IC3 was, unsurprisingly, significantly slower than Uppaal, it scaled similarly, i.e. the runtime increased at a similar rate.
$\K$-induction timed out at three processes already while the booleanization-based approach showed exponential runtime growth and timed out at five processes.

%% file: table.tex
\begin{tabular}{| c | c || c | c | c | c | c || c | c | c | c | c || c | c | c | c | c |}
\hline
 &  & \multicolumn{5}{ c ||}{Full} & \multicolumn{5}{ c ||}{Medium size submodel} & \multicolumn{5}{ c |}{Small submodel}\\
Property & Satisfied & \multicolumn{1}{c}{\begin{sideways}Timed IC3,\end{sideways}\begin{sideways}extended r.\end{sideways}} & \multicolumn{1}{c}{\begin{sideways}Timed IC3,\end{sideways}\begin{sideways}basic regions\end{sideways}} & \multicolumn{1}{c}{\begin{sideways}k-induction,\end{sideways}\begin{sideways}basic regions\end{sideways}} & \multicolumn{1}{c}{\begin{sideways}k-induction,\end{sideways}\begin{sideways}extended r.\end{sideways}} & \multicolumn{1}{c ||}{\begin{sideways}Booleaniza-\end{sideways}\begin{sideways}\ tion + IC3\end{sideways}} & \multicolumn{1}{c}{\begin{sideways}Timed IC3,\end{sideways}\begin{sideways}extended r.\end{sideways}} & \multicolumn{1}{c}{\begin{sideways}Timed IC3,\end{sideways}\begin{sideways}basic regions\end{sideways}} & \multicolumn{1}{c}{\begin{sideways}k-induction,\end{sideways}\begin{sideways}basic regions\end{sideways}} & \multicolumn{1}{c}{\begin{sideways}k-induction,\end{sideways}\begin{sideways}extended r.\end{sideways}} & \multicolumn{1}{c ||}{\begin{sideways}Booleaniza-\end{sideways}\begin{sideways}\ tion + IC3\end{sideways}} & \multicolumn{1}{c}{\begin{sideways}Timed IC3,\end{sideways}\begin{sideways}extended r.\end{sideways}} & \multicolumn{1}{c}{\begin{sideways}Timed IC3,\end{sideways}\begin{sideways}basic regions\end{sideways}} & \multicolumn{1}{c}{\begin{sideways}k-induction,\end{sideways}\begin{sideways}basic regions\end{sideways}} & \multicolumn{1}{c}{\begin{sideways}k-induction,\end{sideways}\begin{sideways}extended r.\end{sideways}} & \multicolumn{1}{c |}{\begin{sideways}Booleaniza-\end{sideways}\begin{sideways}\ tion + IC3\end{sideways}}\\
\hline
\hline
1 & yes & 0.55 & 0.55 & \textbf{0.5} & 0.53 & 181.95 & 0.32 & 0.33 & 0.42 & \textbf{0.29} & 16.43 & 0.21 & 0.21 & 0.29 & \textbf{0.19} & 8.51\\
2 & yes & 0.55 & 0.51 & \textbf{0.47} & 0.5 & \emph{timeout} & \textbf{0.3} & 0.32 & 0.33 & 0.34 & \emph{timeout} & \textbf{0.21} & \textbf{0.21} & 0.23 & 0.22 & 1459.6\\
3 & yes & 0.57 & 0.56 & \textbf{0.49} & \textbf{0.49} & \emph{timeout} & \textbf{0.32} & 0.33 & 0.36 & 0.36 & \emph{timeout} & 0.22 & \textbf{0.21} & 0.23 & 0.24 & \emph{timeout}\\
4 & yes & 0.54 & 0.57 & \textbf{0.48} & 0.58 & \emph{timeout} & \textbf{0.33} & 0.34 & 0.35 & 0.34 & \emph{timeout} &  &  &  &  & \\
5 & yes & 0.68 & 0.55 & 0.62 & \textbf{0.54} & 191.81 & 0.32 & 0.31 & \textbf{0.3} & 0.31 & 18.1 &  &  &  &  & \\
6 & yes & 0.57 & 0.58 & 0.6 & \textbf{0.5} & \emph{timeout} & 0.34 & \textbf{0.33} & 0.37 & 0.38 & \emph{timeout} &  &  &  &  & \\
7 & yes & 0.57 & 0.6 & \textbf{0.51} & 0.55 & \emph{timeout} &  &  &  &  &  &  &  &  &  & \\
\hline
\hline
8 & no & \emph{timeout} & \emph{timeout} & \textbf{2.21} & 2.24 & \emph{timeout} & 0.3 & \textbf{0.29} & 0.35 & 0.33 & \emph{timeout} & 0.21 & \textbf{0.2} & 0.33 & 0.23 & 2367.16\\
9 & no & 0.62 & \textbf{0.56} & 0.7 & 0.65 & 194.25 & 0.32 & 0.31 & 0.29 & \textbf{0.27} & 4.47 & 0.21 & 0.23 & 0.29 & \textbf{0.2} & 2.05\\
10 & no & \emph{timeout} & \emph{timeout} & 2.36 & \textbf{2.15} & 165.16 & \textbf{0.31} & 0.32 & 0.41 & \textbf{0.31} & 17.07 & \textbf{0.2} & 0.23 & 0.21 & \textbf{0.2} & 8.58\\
11 & no & \textbf{0.53} & 0.54 & 0.62 & 0.65 & 169.91 & 0.33 & 0.35 & 0.41 & \textbf{0.32} & 17.53 &  &  &  &  & \\
12 & no & \emph{timeout} & \emph{timeout} & 2.21 & \textbf{2.11} & \emph{timeout} & 0.32 & \textbf{0.31} & 0.34 & 0.33 & \emph{timeout} &  &  &  &  & \\
13 & no & \emph{timeout} & \emph{timeout} & 2.24 & \textbf{2.17} & \emph{timeout} & 0.32 & \textbf{0.31} & 0.46 & 0.35 & \emph{timeout} &  &  &  &  & \\
14 & no & \textbf{0.57} & \textbf{0.57} & 0.63 & 0.65 & 170.87 & \textbf{0.32} & 0.36 & 0.42 & 0.33 & 17.8 &  &  &  &  & \\
15 & no & \textbf{0.56} & 0.59 & 0.85 & 0.66 & 81.07 &  &  &  &  &  &  &  &  &  & \\
\hline
\end{tabular}

%% file: conclusion.tex
This paper introduces two verification methods for symbolic timed transition systems: a timed variant of the IC3 algorithm and an adapted version of $k$-induction. Furthermore, a potential optimization to both methods is devised.

Both methods were able to verify properties on an industrial model
verification of which had been found in previous attempts intractable
and outperformed a booleanization-based approach significantly. Random
properties on the same model revealed that the timed IC3 variant
performs better for satisfied properties while timed $k$-induction
performs better on violated properties. The experiments suggest that
executing timed IC3 in parallel with bounded model checking would yield
excellent performance for the verification of large, non-deterministic
real-time systems.

Additionally, the proposed methods were evaluated on another family of
benchmark, the Fischer mutual exclusion protocol with a varying number
of processes. This family has only a small amount of non-determinism and
the runtime of the methods was higher than that of 
the timed automata model checker Uppaal. However, the timed IC3
algorithm was found to have similar good scaling as Uppaal.